\documentclass[journal]{IEEEtran}

\usepackage{amsmath}
\usepackage{amsfonts}
\usepackage{amssymb}
\usepackage{amsthm}
\usepackage{tabularx}
\usepackage{mathrsfs} 

\usepackage{url}
\usepackage[colorlinks]{hyperref}

\newtheorem{Proposition}{Proposition}

\newtheorem{corollary}{Corollary}

\usepackage{stfloats}
\usepackage{float}
\usepackage{graphicx}
\usepackage{xcolor}
\usepackage{subfigure}

\makeatletter
\def\blfootnote{\xdef\@thefnmark{}\@footnotetext}
\makeatother

\begin{document}

\title{\huge{Fluid Antenna-Aided Rate-Splitting Multiple Access}} 
\author{Farshad~Rostami~Ghadi,~\IEEEmembership{Member},~\textit{IEEE}, 
            Kai-Kit~Wong,~\IEEEmembership{Fellow},~\textit{IEEE}, 
            F.~Javier~L\'{o}pez-Mart\'{i}nez,~\IEEEmembership{Senior~Member},~\textit{IEEE}, 
            Lajos~Hanzo,~\IEEEmembership{Life~Fellow},~\textit{IEEE}, and 
            Chan-Byoung~Chae,~\IEEEmembership{Fellow},~\textit{IEEE}
\vspace{-8mm}
}
\maketitle

\begin{abstract}
This letter considers a fluid antenna system (FAS)-aided rate-splitting multiple access (RSMA) approach for downlink transmission. In particular, a base station (BS) equipped with a single traditional antenna system (TAS) uses RSMA signaling to send information to several mobile users (MUs) each equipped with FAS. To understand the achievable performance, we first present the distribution of the equivalent channel gain based on the joint multivariate $t$-distribution and then derive a compact analytical expression for the outage probability (OP). Moreover, we obtain the asymptotic OP in the high signal-to-noise ratio (SNR) regime. Numerical results show that combining FAS with RSMA significantly outperforms TAS and conventional multiple access schemes, such as non-orthogonal multiple access (NOMA), in terms of OP. The results also indicate that FAS can be the tool that greatly improves the practicality of RSMA.
\end{abstract}


\maketitle


\blfootnote{The work of F. Rostami Ghadi and K. K. Wong is supported by the Engineering and Physical Sciences Research Council (EPSRC) under Grant EP/W026813/1. The work of F. J. L\'opez-Mart\'inez is funded in part by Junta de Andaluc\'ia through grant EMERGIA20-00297, and in part by MICIU/AEI/10.13039/50110001103 and FEDER/UE through grant PID2023-149975OB-I00 (COSTUME). The work of C.-B. Chae is supported by the Institute for Information and Communication Technology Planning and Evaluation (IITP)/NRF grant funded by the Ministry of Science and ICT (MSIT), South Korea, under Grant RS-2024-00428780 and 2022R1A5A1027646.}

\blfootnote{\noindent F. Rostami Ghadi and K. K. Wong are with the Department of Electronic and Electrical Engineering, University College London, London, UK. (e-mail: $\rm \{f.rostamighadi,kai\text{-}kit.wong\}@ucl.ac.uk$). K. K. Wong is also affiliated with Yonsei Frontier Lab, Yonsei University, Seoul, Korea.}
\blfootnote{F. J. L\'opez-Mart\'inez is with the Department of Signal Theory, Networking and Communications, Research Centre for Information and Communication Technologies (CITIC-UGR), University of Granada, 18071, Granada (Spain). (e-mail: $\rm fjlm@ugr.es)$.}
\blfootnote{L. Hanzo is with the School of Electronics and Computer Science, University of Southampton, Southampton, U.K. (e-mail: $\rm lh@ecs.soton.ac.uk$).}
\blfootnote{C. B. Chae is with Yonsei Frontier Lab, Yonsei University, Seoul, Korea. (e-mail: $\rm cbchae@yonsei.ac.kr$).}
	
\blfootnote{Corresponding author: Kai-Kit Wong.}

\vspace{-2mm}
\section{Introduction}\label{sec-intro}
\IEEEPARstart{T}{he rapid} advance of wireless communication systems is driven by the growing demand for sophisticated  multimedia applications. Hence, next-generation networks (NGNs), a.k.a.~the sixth-generation (6G), are aimed to support increased spectral efficiency, massive connectivity, and enhanced quality of service (QoS), including ultra-reliable low-latency communications (URLLC) and real-time data processing \cite{you2023toward}. Among the innovative solutions being explored, such as ultra-massive multiple-input multiple-output (UM-MIMO) systems, reconfigurable intelligent surfaces (RIS), and integrated sensing and communications (ISAC), fluid antenna system (FAS) appears to be the latest breakthrough driven by new reconfigurability in antenna such as shape and position \cite{wong2020fluid,new2024tutorial}.

In contrast to traditional antenna system (TAS), FAS offers a new degree-of-freedom to the physical layer, and therefore can significantly enhance both the diversity and multiplexing gains \cite{Wong-twc2021,new2023information}. By leveraging the unique capabilities of FAS, substantial efforts have recently been devoted to integrating it with cutting-edge technologies, such as MIMO \cite{new2023information,ye2023fluid}, RIS \cite{ghadi2024performance}, and ISAC \cite{wang2024fluid}, among others. Experimental results on FAS have also been reported recently in \cite{Zhang-pFAS2024}.

A key application that FAS may find being very useful is the multiuser communication scenarios where spectrally efficient techniques such as non-orthogonal multiple access (NOMA) and rate-splitting multiple access (RSMA) are utilized. NOMA allows multiple mobile users (MUs) to share the same physical channel to boost its overall spectral efficiency but relies on the use of successive interference cancellation (SIC) at the MUs. Similar to NOMA, RSMA has been piped to do the same or better with a more general and robust transmission framework \cite{mao2018rate}. RSMA splits MUs' data into common and private message components that are transmitted simultaneously, allowing for more effective interference mitigation \cite{clerckx2019rate}. But regardless, interference cancellation is an important part for both NOMA and RSMA, where FAS could shine due to its massive spatial diversity. Note that in \cite{new2023fluid}, FAS has been shown to bring down the power consumption a lot when combined with NOMA.

In this letter, we investigate the integration of FAS with the RSMA scheme, focusing on how this combination can enhance communication reliability and performance in scenarios with multiple MUs and varying channel conditions. Specifically, we examine a multiuser downlink wireless communication system in which MUs employ FAS, while the base station (BS) is equipped with a TAS to adopt the RSMA scheme to simultaneously serve these MUs. We first derive the distribution of the equivalent received channel gain at the FAS-equipped MUs and then characterize the outage probability (OP) in terms of the joint cumulative distribution function (CDF) of the multivariate $t$-distribution using a copula-based formulation. Our analysis reveals that the integration of FAS with RSMA outperforms the configurations where MUs utilize TAS and the BS employs the NOMA signalling. This highlights the advantages of combining fluid antenna technology with innovative multiple access methods to enhance the overall efficiency.

\section{System Model}\label{sec-sys}
\subsection{Channel Model}
Consider the downlink of a FAS-aided RSMA communication setup, which consists of a TAS-equipped BS and $K$ FAS-equipped MUs $\mathrm{u}_k$, for $k=1,\dots,K$. Each MU $\mathrm{u}_k$ features a grid structure consisting of $N_k^l$ ports, uniformly distributed over two linear spaces of length $W_k^l\lambda$ for $l\in\left\{1,2\right\}$, with $\lambda$ being the carrier wavelength. As a result, the total number of ports for MU $\mathrm{u}_k$ is $N_k=N_k^1\times N_k^2$ and the total area of the surface is $W_k=\left(W_k^1\times W_k^2\right)\lambda^2$. To simplify our notations, we introduce a mapping function $\mathcal{F}\left(n_k\right)=\left(n_k^1,n_k^2\right)$, also $n_k=\mathcal{F}^{-1}\left(n_k^1,n_k^2\right)$, which converts the two-dimensional (2D) indices into a one-dimensional (1D) form such that $n_{k}\in\left\{1,\dots,N_k\right\}$ and $n_k^l\in\left\{1,\dots,N_k^l\right\}$. The complex channels from the BS to MU $\mathrm{u}_k$ can be modeled as \cite{new2023fluid}
\begin{multline}
\mathbf{h}_k=\sqrt{\frac{\mathfrak{K}}{\mathfrak{K}+1}}\mathrm{e}^{j\omega}\mathbf{a}\left(\theta_{0,k},\psi_{0,k}\right)\\
+\sqrt{\frac{1}{\mathfrak{L}\left(\mathfrak{K}+1\right)}}\sum_{\mathfrak{l}=1}^\mathfrak{L}\kappa_{\mathfrak{l},k}\mathbf{a}\left(\theta_{\mathfrak{l},k},\psi_{\mathfrak{l},k}\right),
\end{multline}
where $\mathfrak{K}$ represents the Rice factor, $\omega$ denotes the phase of the line-of-sight (LoS) component, $\kappa_{\mathfrak{l},k}$ indicates the complex channel coefficient of the $\mathfrak{l}$-th scattered component, and $\mathfrak{L}$ is the receive steering vector, which is defined as
\begin{align}
&\mathbf{a}\left(\theta_{\mathfrak{l},k},\psi_{\mathfrak{l},k}\right)\notag\\
&=\left[1\quad \mathrm{e}^{j\frac{2\pi W_k^1}{N_k^1-1}\sin\theta_{\mathfrak{l},k}\cos\psi_{\mathfrak{l},k}}\dots \mathrm{e}^{j2\pi W_k^1\sin\theta_{\mathfrak{l},k}\cos\psi_{\mathfrak{l},k}}\right]^T\notag\\
&\otimes \left[1\quad \mathrm{e}^{j\frac{2\pi W_k^2}{N_k^2-1}\sin\theta_{\mathfrak{l},k}\cos\psi_{\mathfrak{l},k}}\dots \mathrm{e}^{j2\pi W_k^2\sin\theta_{\mathfrak{l},k}\cos\psi_{\mathfrak{l},k}}\right],
\end{align}
where $\theta_{\mathfrak{l},k}$ and $\psi_{\mathfrak{l},k}$ denote the azimuth and elevation angle-of-arrival, respectively, and $\otimes$ is the Kronecker tensor product. It is worth mentioning that by assuming  $\mathfrak{L}\rightarrow\infty$ and $\mathfrak{K}=0$, the magnitude of each entry of $\mathbf{h}_k$ follows Rayleigh fading. Therefore, given that the fluid antenna ports can freely switch to any position and be arbitrarily close to each other, the corresponding channels are spatially correlated. More specifically, the covariance between two arbitrary ports $n_k=\mathcal{F}^{-1}\left(n_k^1,n_k^2\right)$ and $\tilde{n}_k=\mathcal{F}^{-1}\left(\tilde{n}_k^1,\tilde{n}_k^2\right)$ at MU $\mathrm{u}_k$ for the considered scenario under rich scattering is given by \cite{new2023information}
\begin{align}
\varrho_{n_k,\tilde{n}_k}=\mathcal{J}_0\left(2\pi\sqrt{\left(\frac{n_k^1-\tilde{n}_k^1}{N^1_k-1}W_k^1\right)^2+\left(\frac{n_k^2-\tilde{n}_k^2}{N^2_k-1}W_k^2\right)^2}\right),
\end{align}
where $\mathcal{J}_0\left(\cdot\right)$ defines the zero-order spherical Bessel function of the first kind.

\subsection{Signal Model}
It is assumed that the BS exploits RSMA signalling for simultaneously supporting $K$ MUs such that messages $w_k$ are transmitted to MUs $\mathrm{u}_k$. Given the RSMA principle, $w_k$ is split into the common message $w_{\mathrm{c}}$ and private message $w_{\mathrm{p},k}$. A shared codebook between MUs is used to jointly encode $w_{\mathrm{c}}$ into a common stream $s_\mathrm{c}$, which must be decoded by all MUs. Simultaneously, $w_{\mathrm{p},k}$ are encoded into the respective private streams $s_\mathrm{p}$. Consequently, the signal transmitted by the BS to the $n$-th port of each MU is defined as
\begin{align}
x^{n_k}=\sqrt{P}\left(\sqrt{\alpha_\mathrm{c}}w_\mathrm{c}+\sum_{k=1}^K\sqrt{\alpha_{\mathrm{p},k}}w_{\mathrm{p},k}\right),
\end{align}
where $P$ denotes the transmit power, while $\alpha_\mathrm{c}$ and $\alpha_{\mathrm{p},k}$ represent the power allocation factors for $s_\mathrm{c}$ and $s_\mathrm{p}$, respectively, so that $\alpha_\mathrm{c}+\sum_{k=1}^K\alpha_{\mathrm{p},k}=1$. Therefore, the received signal at the $n$-th port of the $k$-th MU can be expressed as
\begin{align}
y_k^{n_k}&=h_k^{n_k}x^{n_k}+z^{n_k}_k\notag\\
&=\underset{\text{common message}}{\underbrace{\sqrt{PL_k\alpha_\mathrm{c}}h_k^{n_k}w_\mathrm{c}}}+\underset{\text{private message}}{\underbrace{\sqrt{PL_k\alpha_{\mathrm{p},k}}h_k^{n_k}w_{\mathrm{p},k}}}\notag\\
&\quad\quad+\underset{\text{interference}}{\underbrace{\sum_{\tilde{k}=1\atop \tilde{k}\neq k}^K\sqrt{PL_k\alpha_{\mathrm{p},\tilde{k}}}h_k^{n_k}w_{\mathrm{p},\tilde{k}}}}+z_k^{n_k}, \label{eq-y}
\end{align}
in which $h_k^{n_k}$ is the fading channel coefficient at the $n$-th port of MU $\mathrm{u}_k$, $z_k^{n_k}$ denotes the complex additive white Gaussian noise (AWGN) with zero mean and variance of $\sigma^2$, and $L_k=d_k^{-\beta}$ represents the path-loss, where $\beta>2$ is the path-loss exponent and $d_k$ denotes the Euclidean distance between the BS placed at $\left(X_0,Y_0,Z_0\right)$ and the $k$-th MU located at $\left(X_k,Y_k,Z_k\right)$, which can be determined as 
\begin{align}
d_k = \sqrt{\left(X_0-X_k\right)^2+\left(Y_0-Y_k\right)^2+\left(Z_0-Z_k\right)^2}.
\end{align}

\subsection{Signal-to-Interference Plus Noise Ratio (SINR)}
It is evident from \eqref{eq-y} that in addition to the common and intended private messages, each MU also receives the private messages intended for other MUs. This contributes to the interference level, complicating the decoding of the desired messages. To mitigate this, each MU employs a two-step decoding process for extracting the intended information from the received signal. In the initial step, the MU focuses on decoding the common message, while treating all private messages as noise. Following the FAS concept, we assume that only the optimal port that maximizes the received SINR at the FAS-equipped MUs is activated. Therefore, the SINR for the $k$-th MU can be mathematically expressed as
\begin{align}
\gamma_{\mathrm{c},k}=\frac{\overline{\gamma}\alpha_\mathrm{c}L_k\left|h_k^{n^*_k}\right|^2}{\overline{\gamma}\left(1-\alpha_\mathrm{c}\right)L_k\left|h_k^{n^*_k}\right|^2+1},\label{sinr-c}
\end{align}
in which $\overline{\gamma}=\frac{P}{\sigma^2}$ defines the average transmit signal-to-noise ratio (SNR) and $n_k^*$ denotes the best port index at MU $\mathrm{u}_k$ that maximizes the channel gain, i.e., 
\begin{align}
n_k^*=\underset{n}{\arg\max}\left\{\left|h_k^{n_k}\right|^2\right\}.
\end{align}
Thus, the equivalent channel gain at MU $\mathrm{u}_k$ is given by
\begin{align}
g_{\mathrm{fas},k}=\max\left\{g_k^1,\dots,g_k^{n_k}\right\}, \label{eq-gfas}
\end{align}
where $g_k^{n_k}=\left|h_k^{n_k}\right|^2$.

In the second stage of the RSMA scheme, after successfully decoding the common message, each MU moves on to decode its own private message. This is achieved by subtracting the decoded common message from the received signal, while treating the private messages intended for other MUs as noise. At this point, the MU focuses on isolating its desired private message. Hence, the SINR for decoding the private message at the $k$-th FAS-equipped MU can then be expressed as
\begin{align}
\gamma_{\mathrm{p},k}=\frac{\overline{\gamma}\alpha_{\mathrm{p},k}L_k\left|h_k^{n^*_k}\right|^2}{\overline{\gamma}L_k\left|h_k^{n^*_k}\right|^2\sum_{\tilde{k}=1\atop \tilde{k}\neq k}^K \alpha_{\mathrm{p},k}+1}. \label{sinr-p}
\end{align}

\section{Performance Analysis}
Here, we first derive the CDF and the probability density function (PDF) of the received SINR at the MUs, for the case of Rayleigh fading. We then provide the OP expression, along with its asymptotic form in the high-SNR regime.

\subsection{SINR Distribution}
Given \eqref{eq-gfas}, the CDF of the channel gain at the FAS-RSMA MU $\mathrm{u}_k$ can be mathematically defined as
\begin{align}
F_{g_{\mathrm{fas},k}}\left(g\right)& = \Pr\left(\max\left\{g_k^1,\dots,g_k^{n_k}\right\}\leq g\right)\\
&=\Pr\left(g_k^1\leq g,\dots, g_k^{n_k}\leq g\right)\\
&=F_{g_k^1,\dots,g_k^{n_k}}\left(g,\dots,g\right), \label{eq-cdf1}
\end{align}
where \eqref{eq-cdf1} is the joint multivariate CDF of $N_k$ correlated exponential random variables (RVs) due to the spatial correlation between the corresponding Rayleigh channel coefficients. It is evident that deriving the closed-from expression of $F_{g_{\mathrm{fas},k}}\left(g\right)$ is mathematically intractable. Nevertheless, the joint CDF of arbitrary multivariate RVs can be efficiently generated by using Sklar's theorem such that for $d$ arbitrary correlated RVs $S_i$, $i=1,\dots,d$, associated with the univariate marginal CDF $F_{S_i}\left(s_i\right)$, the  joint multivariate CDF $F_{S_1,\dots,S_d}\left(s_1,\dots,s_d\right)$,  in the extended real line domain $\overline{\mathbb{R}}$, is given by \cite{ghadi2020copula}
\begin{align}
F_{S_1,\dots,S_d}\left(s_1,\dots,s_d\right) = C\left(F_{S_1}\left(s_1\right),\dots,F_{S_d}\left(s_d\right);\Theta\right). \label{eq-sklar}
\end{align}
In \eqref{eq-sklar},  $C(\cdot):\left[0,1\right]^d\rightarrow\left[0,1\right]$ denotes the copula function
that is a joint CDF of $d$ random vectors on the unit cube $\left[0,1\right]^d$ associated with uniform marginal distributions, i.e.,
\begin{align}
C\left(u_1,\dots,u_d;\Theta\right)=\Pr\left(U_1\leq u_1,\dots,U_d\leq u_d\right).
\end{align}
in which $u_i=F_{S_i}\left(s_i\right)$ and $\Theta$ denotes the dependence parameter, which characterizes the dependency between the respective correlated RVs. Following \cite{ghadi2020copula}, it was revealed in \cite{ghadi2024gaussian} that the elliptical copula  accurately captures the spatial correlation in FAS, where the corresponding CDF can be expressed in terms of the Gaussian copula, i.e., the joint multivariate CDF of normal distribution. Although the Gaussian copula is capable of accurately describing the structure of dependency between the fading coefficients, it cannot cover heavy tail dependencies, where deep fading occurs. To overcome this issue, we utilize a more general elliptical student-$t$ copula, which includes an additional parameter $\nu_k$ in addition to the correlation matrix, to control the degree of tail dependencies.  

\begin{Proposition}\label{pro-dist}
The CDF and PDF of the equivalent channel gain $g_{\mathrm{fas},k}$ at MU $\mathrm{u}_k$ for FAS-RSMA are given by
\begin{align}
&F_{g_{\mathrm{fas},k}}\left(g\right)=\notag\\
& T_{\nu_k,\mathbf{\Sigma}_k}\left(t_{\nu_k}^{-1}\left(1-\mathrm{e}^{-\frac{g}{\overline{g}}}\right),\dots,t_{\nu_k}^{-1}\left(1-\mathrm{e}^{-\frac{g}{\overline{g}}}\right);\nu_k,\Theta_k\right), \label{eq-cdf}
\end{align}
and
\begin{align}
&f_{g_{\mathrm{fas},k}}\left(g\right) =\notag\\
& \frac{\Gamma\left(\frac{\nu_k+{N_k}}{2}\right)}{\Gamma\left(\frac{\nu_k}{2}\right)\sqrt{\left(\pi\nu_k\right)^{N_k}\left|\mathbf{\Sigma}_k\right|}}\left(1+\frac{1}{\nu_k}\left(\mathbf{t}^{-1}_{\nu_k}\right)^T\mathbf{\Sigma}_k^{-1}\mathbf{t}^{-1}_{\nu_k}\right)^{-\frac{\nu_k+N_k}{2}}, \label{eq-pdf}
\end{align}
where $\overline{g}=\mathbb{E}\left[g\right]$ is the mean of the channel gain,  $t_{\nu_k}^{-1}\left(\cdot\right)$ is the inverse CDF (quantile function) of the univariate $t$-distribution having $\nu_k$ degrees of freedom for the $k$-th MU,  $T_{\nu_k,\mathbf{\Sigma}_k}\left(\cdot\right)$ represents the CDF of the multivariate $t$-distribution with correlation matrix $\mathbf{\Sigma}_k$ and $\nu_k$ degrees of freedom for the $k$-th MU, and $\Theta_k\in\left[-1,1\right]$ denotes the dependence parameter of the $t$-student copula, which represents the correlation between the $n_k$-th and $\tilde{n}_k$-th ports in $\mathbf{\Sigma}_k$. Moreover, $\Gamma\left(.\right)$ is the Gamma function, $\left|\mathbf{\Sigma}_k\right|$ defines the determinant of $\mathbf{\Sigma}_k$, and
\begin{align}
\mathbf{t}^{-1}_{\nu_k} = \left[t^{-1}_{\nu_k}\left(1-\mathrm{e}^{-\frac{g}{\overline{g}}}\right),\dots,t^{-1}_{\nu_k}\left(1-\mathrm{e}^{-\frac{g}{\overline{g}}}\right)\right].
\end{align}
\end{Proposition}

\begin{proof}
See Appendix \ref{app-pro-dist}.
\end{proof}

\subsection{OP Analysis}
The OP is a key performance metric in wireless communications, defined as the probability that the received SNR or SINR falls below a critical threshold, causing a disruption of reliable communication. In the RSMA-based signaling model, since each MU receives a combination of a common message and its own private message, along with the private messages intended for other MUs, and decodes both messages through a two-step decoding process, an outage occurs when the SINRs of decoding either the common or private message drop below their respective thresholds denoted as $\gamma_\mathrm{th,c}$ for the common message and $\gamma_\mathrm{th,p}$ for the private message.

\begin{Proposition}
The OP of MU $k$ in FAS-RSMA is given by
\begin{align}
&P_{\mathrm{o},k}=\notag\\
& T_{\nu_k,\mathbf{\Sigma}_k}\left(t_{\nu_k}^{-1}\left(1-e^{-\frac{\gamma_\mathrm{th}^k}{\overline{g}}}\right),\dots,t_{\nu_k}^{-1}\left(1-e^{-\frac{\gamma_\mathrm{th}^k}{\overline{g}}}\right);\nu_k,\Theta_k\right),\label{eq-op}
\end{align}
where $\gamma_\mathrm{th}^k=\max\left\{\hat{\gamma}_\mathrm{th,c}^k,\hat{\gamma}_\mathrm{th,p}^k\right\}$ in which
\begin{align}\notag
\hat{\gamma}_\mathrm{th,c}^k = \frac{\gamma_{\mathrm{th,c}}^k}{\overline{\gamma}L_k\left(\alpha_\mathrm{c}-\left(1-\alpha_\mathrm{c}\right)\gamma_\mathrm{th,c}^k\right)}
\end{align}
and
\begin{align}
\hat{\gamma}_\mathrm{th,p}^k = \frac{\gamma_{\mathrm{th,p}}^k}{\overline{\gamma}L_k\left(\alpha_{\mathrm{p},k}-\left(1-\alpha_\mathrm{c}-\alpha_{\mathrm{p},k}\right)\gamma_\mathrm{th,p}^k\right)}.
\end{align}
\end{Proposition}

\begin{proof}
Given the definition, the OP at the $k$-th FAS-equipped RSMA MU can be expressed mathematically as
\begin{subequations}
\begin{align}
P_{\mathrm{o},k}& = 1-\Pr\left(\gamma_{\mathrm{c},k}>\gamma_\mathrm{th,c}^k,\gamma_{\mathrm{p},k}>\gamma_\mathrm{th,p}^k\right) \label{p-a}\\ \notag
&=1-\Pr\Bigg(\frac{\overline{\gamma}\alpha_\mathrm{c}L_kg_{\mathrm{fas},k}}{\overline{\gamma}\left(1-\alpha_\mathrm{c}\right)L_kg_{\mathrm{fas},k}+1}>\gamma_\mathrm{th,c}^k,\label{p-b}\\
&\quad\quad\quad\frac{\overline{\gamma}\alpha_{\mathrm{p},k}L_kg_{\mathrm{fas},k}}{\overline{\gamma}L_kg_{\mathrm{fas},k}\sum_{\tilde{k}=1, \tilde{k}\neq k}^K \alpha_{\mathrm{p},k}+1}>\gamma_\mathrm{th,p}^k\Bigg)\\
& = 1-\Pr\left(g_{\mathrm{fas},k}>\hat{\gamma}_\mathrm{th,c}^k,g_{\mathrm{fas},k}>\hat{\gamma}_\mathrm{th,p}^k\right)\\
& = \Pr\left(g_{\mathrm{fas},k}\leq \max\left\{\hat{\gamma}_\mathrm{th,c}^k,\hat{\gamma}_\mathrm{th,p}^k\right\}\right)\\
& = F_{g_{\mathrm{fas},k}}\left(\gamma_\mathrm{th}^k\right),\label{p-e}
\end{align}
\end{subequations}
where \eqref{p-b} is derived by substituting \eqref{sinr-c} and \eqref{sinr-p} into \eqref{p-a}, while \eqref{p-e} is derived by using \eqref{eq-cdf}.
\end{proof}

\begin{corollary}
The asymptotic OP in the high-SNR regime for the considered FAS-RSMA is given by
\begin{align}
P_{\mathrm{o},k}^\infty=T_{\nu_k,\mathbf{\Sigma}_k}\left(t_{\nu_k}^{-1}\left(\frac{\gamma_\mathrm{th}}{\overline{g}}\right),\dots,t_{\nu_k}^{-1}\left(\frac{\gamma_\mathrm{th}}{\overline{g}}\right);\nu_k,\Theta_k\right). \label{eq-op-asym}
\end{align}
\end{corollary}

\begin{proof}
Given that the channel gains are distributed exponentially, we can derive their respective CDFs at high SNR by employing the Taylor series expansion. Hence, when  $\overline{\gamma}\rightarrow\infty$, the term $\mathrm{e}^{\gamma_\mathrm{th}/\overline{g}}$ becomes very small, so we have $F_{g_k^{n_k}}^\infty\left(\gamma_\mathrm{th}\right)\approx \gamma_\mathrm{th}/\overline{g}$. Consequently, by inserting the obtained marginal CDF $F_{g_k^{n_k}}^\infty\left(\gamma_\mathrm{th}\right)$ into \eqref{eq-op}, the proof is completed.
\end{proof}


\section{Numerical Results}\label{sec-num}
Here, we evaluate the OP performance of FAS-RSMA. Unless otherwise specified, we set the system parameters to $\beta=2.1$, $\alpha_\mathrm{c}=0.7$, $W_k=1\lambda^2$, $N_k=4$, $\nu=40$, and $\gamma_{\mathrm{th,c}}^k=\gamma_{\mathrm{th,p}}^k=0$ dB. For the sake of simplicity in the simulations, we consider a two-user FAS-RSMA, i.e., $K=2$, where the BS is placed at the origin $\left(0,0,0\right)$, and MUs $\mathrm{u_1}$ and $\mathrm{u_2}$ are positioned at $\left(50,50,0\right)$m and $\left(10,50,0\right)$m, respectively. The corresponding power allocation factors for the private messages are $\alpha_{\mathrm{p},1}=0.75\left(1-\alpha_{\mathrm{c}}\right)$ and $\alpha_{\mathrm{p},2}=0.25\left(1-\alpha_{\mathrm{c}}\right)$. Furthermore, we use the MATLAB function \texttt{copulacdf} to implement the $t$-student copula, which is expressed in terms of the joint CDF of the multivariate $t$-distribution. We also consider the following schemes as the benchmark:
\begin{itemize}
\item \textbf{FAS-NOMA:} MUs employ a FAS-aided NOMA.
\item \textbf{TAS-RSMA:} MUs rely on a TAS-aided RSMA.
\item \textbf{TAS-NOMA:} MUs adopt a TAS-aided NOMA.
\end{itemize}

Fig.~\ref{fig_out} shows the OP results against the average transmit SNR $\overline{\gamma}$ for various schemes, including FAS-RSMA, FAS-NOMA, TAS-RSMA, and TAS-NOMA. More specifically, in Fig.~\ref{fig_out_n}, the OP is analyzed for varying the numbers of fluid antenna ports $N_k$, while maintaining a constant fluid antenna area of $W_k=1\lambda^2$. Meanwhile, Fig.~\ref{fig_out_w} examines the OP for different fluid antenna sizes, keeping the number of ports fixed at $N_k=9$. Firstly, we can observe that the asymptotic results closely align with the OP curves in the high-SNR regime, confirming the accuracy of our theoretical analysis for both MUs. We can also see that as $\overline{\gamma}$ increases, the OP decreases for all schemes, which aligns with our expectations, meaning that improved signal quality results in reduced outages. 

\begin{figure}\vspace{0cm}
\centering
\hspace{0cm}\subfigure[]{%
\includegraphics[width=0.85\columnwidth]{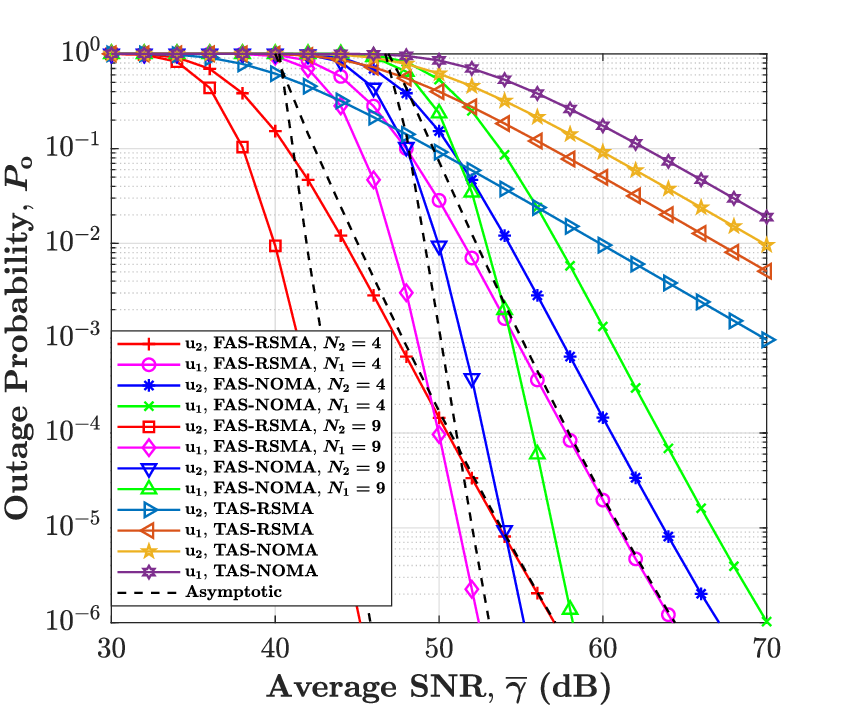}\label{fig_out_n}%
}\vspace{-4mm}\hspace{0cm}
\subfigure[]{%
\includegraphics[width=0.85\columnwidth]{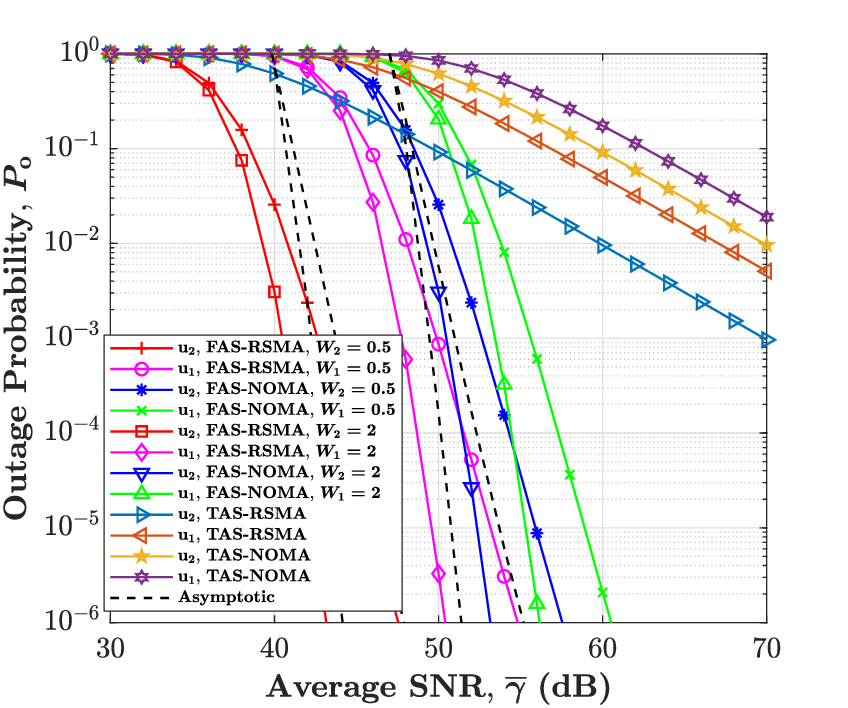}\label{fig_out_w}%
}\hspace{-0.1cm}\vspace{-2mm}
\caption{OP vs. the average transmit SNR $\overline{\gamma}$: (a) for different numbers of fluid antenna ports $N_k$ with a fixed $W_k=1\lambda^2$, and (b) for different values of fluid antenna size with a fixed $N_k=9$.}\label{fig_1}\vspace{-0.3cm}\label{fig_out}
\end{figure}

Among the different configurations, it can be observed that FAS-RSMA consistently stands out as being most effective. For instance in Fig.~\ref{fig_out_n}, at $\overline{\gamma}=56$ dB, FAS-RSMA achieves an OP of approximately $10^{-6}$ for MU $\mathrm{u}_2$ along with $N_2=4$, while under the same setup, the OP of other schemes is much higher. This superior performance can be attributed to the rate-splitting mechanism, which divides messages into common and private components. Clearly, the RSMA scheme allows for a more efficient allocation of resources, allowing the system to optimize communication for both stronger and weaker MUs. Moreover, the FAS enhances adaptability by dynamically adjusting to environmental conditions, thereby improving the signal reception and minimizing interference. As such, FAS-RSMA effectively manages the inter-user interference, allowing robust MUs to assist weaker ones, which further contributes to the reduced OP. In contrast, FAS-NOMA, while also benefiting from a FAS, does not achieve the same performance level as FAS-RSMA. The reliance on power domain multiplexing in NOMA can lead to significant interference, particularly when the MUs experience time-varying channel conditions. This inefficiency in resource utilization hampers FAS-NOMA's robustness in the high SNR regime, making it less effective compared to its RSMA counterpart.

Besides, it can also be observed that TAS-RSMA, which employs a traditional fixed-position antenna system, shows a higher OP compared to FAS-RSMA. This difference highlights the limitations of fixed antenna configurations to leverage spatial diversity in rich-scattering channel conditions. Although TAS-RSMA also utilizes the rate-splitting technique, its performance is hindered by the inability to dynamically allocate resources based on changing channel conditions. While TAS-RSMA benefits from the rate-splitting technique, it is less effective in managing interference. For instance, as observed in Fig.~\ref{fig_out_w}, at $\overline{\gamma}=50$dB, the performance gap between FAS-RSMA and TAS-RSMA is remarkable, with FAS-RSMA achieving an OP of about $10^{-6}$ compared to TAS-RSMA's $10^{-2}$. This demonstrates that FAS-RSMA maintains better performance even at lower SNR levels, thanks to the capability of FAS and its ability to respond to environmental changes.

\section{Conclusion}\label{sec-con}
In this letter, we explored the performance of FAS-aided RSMA in downlink multiuser communications. First of all, we derived an analytical expression for the OP based on the joint CDF of the multivariate $t$-distribution and examined the system's asymptotic behavior in high-SNR conditions. Our findings demonstrated that FAS combined with RSMA delivers notable improvements over TAS and NOMA. These results highlight the effectiveness of integrating fluid antenna technology with advanced communication strategies for significantly enhancing system reliability and performance.

\appendices
\section{Proof of Proposition \ref{pro-dist}}\label{app-pro-dist}
By applying Sklar's theorem from \eqref{eq-sklar} to the CDF definition in \eqref{eq-cdf1}, $F_{g_{\mathrm{fas},k}}\left(g\right)$ can be rewritten as
\begin{align}
F_{g_{\mathrm{fas},k}}\left(g\right) = C\left(F_{g_k^1}\left(g\right),\dots,F_{g_k^{n_k}}\left(g\right);\Theta_k\right), \label{eq-app1}
\end{align}
where $C\left(\cdot\right)$ can be any arbitrary copula and $F_{g_k^{n_k}}\left(g\right)$ is the marginal CDF of the channel gain, which follows exponential distribution due to Rayleigh fading channels, i.e., we have $F_{g_k^{n_k}}\left(g\right)=1-\mathrm{e}^{-g/\overline{g}}$, 
where $\overline{g}=\mathbb{E}\left[g\right]$ is the mean of the channel gain. Following the insights provided in \cite{ghadi2024gaussian}, revealing the fact that elliptical copula can accurately approximate Jakes' model in FAS, we exploit the $t$-student copula, which captures positive/negative and linear/non-linear correlations as well as heavy tail dependencies and it is defined as \cite{sundaresan2011copula}
\begin{align}
C\left(u_1,\dots,u_d\right) = T_{\nu,\mathbf{\Sigma}}\left(t^{-1}_\nu\left(u_1\right),\dots,t^{-1}_\nu\left(u_d\right);\nu,\Theta\right), \label{eq-student}
\end{align}
in which $t_{\nu}^{-1}\left(\cdot\right)$ denotes the quantile function of the univariate $t$-distribution with $\nu$ degrees of freedom,  $T_{\nu,\mathbf{\Sigma}}\left(\cdot\right)$ is the CDF of the multivariate $t$-distribution with correlation matrix $\mathbf{\Sigma}$ and $\nu$ degrees of freedom, and $\Theta\in\left[-1,1\right]$ represents the dependence parameter of $t$-student copula, which measures the correlation between two arbitrary RVs. It should be noted that the $t$-student copula is a more general instance of the elliptical copula, which approaches the Gaussian copula as the degrees of freedom $\nu$ becomes large, i.e., $\nu\rightarrow\infty$. Moreover, it was shown in \cite{ghadi2024gaussian} that the dependence parameter of the elliptical copula such as the Gaussian copula is almost equal to the correlation coefficient derived by Jakes' model, i.e., $\Theta_k\approx \varrho_{n_k,\tilde{n}_k}$. Hence, upon applying \eqref{eq-student} to \eqref{eq-app1}, we have
\begin{align}\notag
&F_{g_{\mathrm{fas},k}}\left(g\right) =\\
& T_{\nu_k,\mathbf{R}_k}\left(t^{-1}\left(F_{g_k^1}\left(g\right)\right),\dots,t^{-1}\left(F_{g_k^{n_k}}\left(g\right)\right);\nu_k,\Theta_k\right). \label{eq-ref3}
\end{align}
Finally, by inserting $F_{g_k^{n_k}}\left(g\right)$ into \eqref{eq-ref3}, \eqref{eq-cdf} is derived. Also, by using the definition of the copula PDF and applying the chain rule to \eqref{eq-student} \cite{ghadi2020copula}, \eqref{eq-pdf} is therefore obtained.


\end{document}